\documentclass{article}
\usepackage{yuzhao}

\usepackage{cleveref}

\usepackage{multirow}

\newcommand{\kWISE}[1]{\scalebox{.8}[.8]{\textnormal{{#1}WISE}}}

\begin{document}
	
	\title{On closeness to $k$-wise uniformity}
	
	\author{Ryan O'Donnell\thanks{Supported by NSF grants CCF-1618679, CCF-1717606. This material is based upon work supported by the National Science Foundation under grant numbers listed above. Any opinions, findings and conclusions or recommendations expressed in this material are those of the author and do not necessarily reflect the views of the National Science Foundation (NSF).} \and Yu Zhao$^*$}
	
	\date{\vspace{-5ex}}
	
\maketitle

\begin{abstract}
	A probability distribution over $\{-1, 1\}^n$ is \emph{$(\eps, k)$-wise uniform} if, roughly, it is $\eps$-close to the uniform distribution when restricted to any $k$ coordinates. We consider the problem of how far an $(\epsilon, k)$-wise uniform distribution can be from any globally $k$-wise uniform distribution. We show that every $(\epsilon, k)$-wise uniform distribution is $O(n^{k/2}\epsilon)$-close to a $k$-wise uniform distribution in total variation distance. In addition, we show that this bound is optimal for all even $k$: we find an $(\eps, k)$-wise uniform distribution that is $\Omega(n^{k/2}\epsilon)$-far from any $k$-wise uniform distribution in total variation distance. For $k=1$, we get a better upper bound of $O(\eps)$, which is also optimal.
	
	One application of our closeness result is to the sample complexity of testing whether a distribution is $k$-wise uniform or $\delta$-far from $k$-wise uniform. We give an upper bound of $O(n^{k}/\delta^2)$ (or $O(\log n/\delta^2)$ when $k = 1$) on the required samples. We show an improved upper bound of $\tilde{O}(n^{k/2}/\delta^2)$ for the special case of testing fully uniform vs.\ $\delta$-far from $k$-wise uniform. Finally, we complement this with a matching lower bound of $\Omega(n/\delta^2)$ when $k = 2$.
	
	Our results improve upon the best known bounds from \cite{AAKMRX07}, and have simpler proofs.
\end{abstract}

\section{Introduction}

\subsection{$k$-wise uniformity and almost $k$-wise uniformity}

We say that a probability distribution over $\{-1,1\}^n$ is \emph{$k$-wise uniform} if its marginal distribution on every subset of $k$ coordinates is the uniform distribution. For Fourier analysis of the Hamming cube, it is convenient to identify the distribution with its density function $\varphi : \{-1, 1\}^n \to \R^{\geq 0}$ satisfying
\[
\E_{\bx \sim \{-1, 1\}^n}[\varphi(\bx)] = 1.
\]
We write $\bx \sim \varphi$ to denote that $\bx$ is a random variable drawn from the associated distribution with density~$\varphi$:
\[
\Pr_{\bx \sim \varphi}[\bx = x] = \frac{\varphi(x)}{2^n}
\]
for any $x \in \{-1, 1\}^n$. Then a well-known fact is that a distribution is $k$-wise uniform if and only if the Fourier coefficient of $\varphi$ is $0$ on every subset $S \subseteq [n]$ of size between $1$ and $k$:
\[
\widehat{\varphi}(S) = \E_{\bx \sim \varphi}\left[\prod_{i \in S} \bx_i\right] = 0.
\]

$k$-wise uniformity is an essential tool in theoretical computer science. Its study dates back to work of Rao \cite{Rao47}. He studied $k$-wise uniform sets, which are special cases of $k$-wise uniform distribution. A subset of $\{-1, 1\}^n$ is a  \emph{$k$-wise uniform set} if the uniform distribution on this subset is $k$-wise uniform. Rao gave constructions of a pairwise-uniform set of size $n+1$ (when $n=2^r-1$ for any integer $r$), a $3$-wise uniform set of size $2n$ (when $n = 2^r$ for any integer $r$), and a lower bound (reproved in \cite{ABI86, CGHFRS85}) that a $k$-wise uniform set on $\{-1,1\}^n$ requires size at least $\Omega(n^{\lfloor k/2 \rfloor})$. An alternative proof of the lower bound for even~$k$ is shown in \cite{AGM03} using a hypercontractivity-type technique, as opposed to the linear algebra method. Coding theorists have also heavily studied $k$-wise uniformity, since MacWilliams and Sloane showed that linear codes with dual minimum distance $k+1$ correspond to $k$-wise uniform sets in \cite{MS77}. The importance in theoretical computer science of $k$-wise independence for derandomization arose simultaneously in many papers, with \cite{KW85,Luby86} emphasizing derandomization via the most common pairwise-uniformity case, and \cite{ABI86,CGHFRS85} emphasizing derandomization based on $k$-wise independence more generally.

A distribution is ``almost $k$-wise uniform'' if its marginal distribution on every $k$ coordinates is very close to the uniform distribution. Typically we say two distributions $\varphi, \psi$ are \emph{$\delta$-close}, if the total variation distance between $\varphi$ and $\psi$ is at most $\delta$; and we say they are \emph{$\delta$-far}, if the total variation distance between them is more than $\delta$. However the precise notion of ``close to uniform'' has varied in previous work. Suppose $\psi$ is the density function for the marginal distribution of $\varphi$ restricted to some specific $k$ coordinates and $\bone$ is the density function for the uniform distribution. Several standard ways are introduced in \cite{AGM03,AAKMRX07} to quantify closeness to uniformity, corresponding to the $L_1, L_2,  L_\infty$ norms:
\begin{itemize}
	\item ($L_1$ norm): $\|\psi - \bone\|_1 = 2d_{\text{TV}}(\psi, \bone) \leq \eps$, where $d_{\text{TV}}$ denotes total variation distance;
	\item ($L_2$ norm): $\|\psi - \bone \|_2 = \sqrt{\chi^2(\psi, \bone)} = \sqrt{\sum_{S \neq \emptyset} \widehat{\psi}(S)^2}  \leq \eps$, where $\chi^2(\psi, \bone)$ denotes the $\chi^2$-divergence of $\psi$ from the uniform distribution;
	\item ($L_\infty$ norm): $\|\psi - \bone\|_\infty \leq \eps$, or in other words, for any $x \in \{-1, 1\}^n$,
	\[
	\left|\Pr_{\bx \sim \psi}[\bx = x] - 2^{-k}\right| \leq 2^{-k} \eps.
	\]
\end{itemize}

Note the following: First, closeness in $L_1$ norm is the most natural for algorithmic derandomization purposes: it tells us that the algorithm cannot tell $\psi$ is different from the uniform distribution up to $\eps$ error. Second, these definitions of closeness are in increasing order of strength. On the other hand, we also have that $\| \psi-\bone\|_1 \leq \|\psi - \bone \|_\infty \leq 2^k \| \psi-\bone\|_1$; thus all of these notions are within a factor of $2^k$. We generally consider $k$ to be constant (or at worst, $O(\log n)$), so that these notions are roughly the same. 

A fourth reasonable notion, proposed by Naor and Naor in \cite{NN93}, is that the distribution has a small bias over every non-empty subset of at most $k$ coordinates. We say density function $\varphi$ is \emph{$(\eps, k)$-wise uniform} if for every non-empty set $S \subseteq [n]$ with size at most $k$,
\[
|\widehat{\varphi}(S)| = \left|\Pr_{\bx \sim \varphi}\left[\prod_{i \in S}\bx_i = 1\right] - \Pr_{\bx \sim \varphi}\left[\prod_{i \in S}\bx_i = -1\right]\right| \leq \epsilon.
\]

Here we also have $\eps = 0$ if and only if $\varphi$ is exactly $k$-wise uniform. Clearly if the marginal density of $\varphi$ over every $k$ coordinates is $\epsilon$-close to the uniform distribution in total variation distance, then $\varphi$ is~$(\eps, k)$-wise uniform. On the other hand, if $\varphi$ is~$(\epsilon,k)$-wise uniform, then the marginal density of $\varphi$ over every~$k$ coordinates is $2^{k/2}\eps$-close to uniform distribution in total variation distance. Again, if $k$ is considered constant, this bias notion is also roughly the same as previous notions. In the rest of paper we prefer this~$(\epsilon, k)$-wise uniform notion for ``almost $k$-wise uniform'' because of its convenience for Fourier analysis.

The original paper about almost $k$-wise uniformity, \cite{NN93}, was concerned with derandomization; e.g., they use $(\epsilon, k)$-wise uniformity for derandomizing the ``set balancing (discrepancy)'' problem. Alon et al.\ gave a further discussion of the relationship between almost $k$-wise uniformity and derandomization in~\cite{AGM03}. The key idea is the following: In many cases of randomized algorithms, the analysis only relies on the property that the random bits are $k$-wise uniform, as opposed to fully uniform. Since there exists an efficiently samplable $k$-wise uniform distribution on a set of size at most $O(n^{\lfloor k/2 \rfloor})$, one can reduce the number of random unbiased bits used in the algorithm down to $O(k \log n)$. To further reduce the number of random bits used, a natural line of thinking is to consider distributions which are ``almost $k$-wise uniformity''. Alon et al.\ \cite{AGHP92} showed that we can deterministically construct $(\epsilon, k)$-wise uniform sets that are of size~$\text{poly}(2^k,\log n,1/ \eps)$, much smaller than exact $k$-wise uniform ones (roughly $\Omega(n^{\lfloor k/2 \rfloor})$ size). Therefore we can use substantially fewer random bits by taking random strings from an almost $k$-wise uniform distribution. 

However we need to ensure that the original analysis of the randomized algorithm still holds under the almost $k$-wise uniform distribution. This is to say that if the randomized algorithm behaves well on a $k$-wise uniform distribution, it may or may not also work as well with an $(\epsilon, k)$-wise uniform distribution, when the parameter~$\epsilon$ is small enough. 

\subsection{The Closeness Problem}

For the analysis of derandomization, it would be very convenient if $(\epsilon, k)$-wise uniformity -- which means that ``every $k$-local view looks close to uniform'' -- implies global $\delta$-closeness to $k$-wise uniformity. A natural question that arises, posed in \cite{AGM03}, is the following:

\medskip

\emph{How small can $\delta$ be such that the following is true? For every $(\eps, k)$-wise uniform distribution $\varphi$ on $\{-1, 1\}^n$, $\varphi$ is $\delta$-close to some $k$-wise uniform distribution.}

\medskip

In this paper, we will refer to this question as \emph{the Closeness Problem}.

\subsubsection{Previous work and applications}

On one hand, the main message of \cite{AGM03} is a lower bound: For every even constant $k > 4$, they gave an $(\eps,k)$-wise uniform distribution with $\eps = O (1/n^{k/4 - 1})$, yet which is $\frac12$-far from every $k$-wise uniform distribution in total variation distance. 

On the other hand, \cite{AGM03} proved a very simple theorem that $\delta \leq O(n^k \eps)$ always holds. Despite its simplicity, this upper bound has been used many times in well known results. 

One application is in circuit complexity. \cite{AGM03}'s upper bound is used for fooling disjunctive normal formulas (DNF) \cite{Bazzi09} and $\AC^0$ \cite{Braverman10}. In these works, once the authors showed that $k$-wise uniformity suffices to fool DNF/$\AC^0$, they deduced that $(O(1/n^k),k)$-uniform distributions suffice, and hence $O(1/n^k)$-biased sets sufficed trivially. \cite{AGM03}'s upper bound is also used as a tool for the construction of two-source extractors for a similar reason in \cite{CZ16,Li16}.  

Another application is for hardness of constraint satisfactory problems ($\CSP$s).  Austrin and Mossel \cite{AM09} show that one can obtain integrality gaps and UGC-hardness for CSPs based on $k$-wise uniform distributions of small support size.  If a predicate is $k$-wise uniform, Kothari et al.\ \cite{KMOW17} showed that one can get SOS-hardness of refuting random instances of it when there are around $n^{(k+1)/2}$ constraints.  Indeed, \cite{KMOW17} shows that if we have a predicate that is $\delta$-close to $k$-wise uniform, then with roughly $n^{(k+1)/2}$ random constraints, SOS cannot refute that a $(1-O(\delta))$-fraction of constraints are satisfiable. This also motivates studying $\delta$-closeness to $k$-wise uniformity and how it relates to Fourier coefficients. $\delta$-closeness to $k$-wise uniformity is also relevant for hardness of random $\CSP$, as shown in \cite{AOW15}.

Alon et al.\ \cite{AAKMRX07} investigated the Closeness Problem further by improving the upper bound to $\delta = O((n\log n)^{k/2} \eps)$.  Indeed, they showed a strictly stronger fact that a distribution is $O\!\left(\sqrt{\W{1\dots k}[\varphi]} \log^{k/2} n\right)$-close to some $k$-wise uniform, where $\W{1 \dots k}[\varphi] = \sum_{1 \leq |S| \leq k} \wh{\varphi}(S)^2$. Rubinfeld and Xie \cite{RN13} generalized some of these results to non-uniform $k$-wise independent distributions over larger product spaces.

Let us briefly summarize the method \cite{AAKMRX07} used to prove their upper bounds. Given an $(\epsilon, k)$-wise uniform $\varphi$, they first try to generate a $k$-wise uniform ``pseudo-distribution'' $\varphi'$ by forcing all Fourier coefficients at degree at most $k$ to be zero. It is a ``pseudo-distribution'' because some points might have negative density. After this, they use a fully uniform distribution and $k$-wise uniform distributions with small support size to try to mend all points to be nonnegative. They bound the weight of these mending distributions to upper-bound the distance incurred by the mending process. This mending process uses the fully uniform distribution to mend the small negative weights and uses $k$-wise uniform distributions with small support size to correct the large negative weights point by point. By optimizing the threshold between small and large weights it introduces a factor of $(\log n)^{k/2}$.

Though they did not mention it explicitly, they also give a lower bound for the Closeness Problem of~$\delta  \geq \Omega\left(\frac{n^{(k-1)/2}}{\log n} \eps\right)$ for $k > 2$ by considering the uniform distribution on a set of $O(n^k)$ random chosen strings. No previous work gave any lower bound for the most natural case of $k = 2$. 

\subsubsection{Our result}

In this paper, we show sharper upper and lower bounds for the Closeness Problem, which are tight for $k$ even and $k = 1$.  Comparing to the result in \cite{AAKMRX07}, we get rid of the factor of $(\log n)^{k/2}$.

\begin{theorem}
	\label{thm:section1_closeness_upperbound}
	Any density $\varphi$ over $\{-1, 1\}^n$ is $\delta$-close to some $k$-wise uniform distribution, where
	\[
	\delta \leq e^k \sqrt{\W{1 \dots k}[\varphi]} = e^k \sqrt{\sum_{1 \leq |S| \leq k} \widehat{\varphi}(S)^2}.
	\]
	Consequently, if $\varphi$ is $(\epsilon, k)$-wise uniform, i.e., $|\widehat{\varphi}(S)| \leq \epsilon$ for every non-empty set $S$ with size at most $k$, then
	\[
	\delta \leq e^k n^{k/2} \eps.
	\]
	
	For the special case $k = 1$, the corresponding $\delta$ can be further improved to $\delta \leq \epsilon$.
\end{theorem}

Our new technique is trying to mend the original distribution to be $k$-wise uniform all at once. We want to show that some mixture distribution $(\varphi + w \psi)$ is $k$-wise uniform with small mixture weight $w$. The distance between the final mixture distribution and the original distribution $\varphi$ is bounded by $O(w)$. Therefore we only need to show that the mending distribution $\psi$ exists for some small weight $w$. Showing the existence of such a distribution~$\psi$ can be written as the feasibility of a linear program (LP). We upper bound $w$ by bounding the dual LP, using the hypercontractivity inequality.

Our result is sharp for all even $k$, and is also sharp for $k = 1$. We state the matching lower bound for even $k$:
\begin{theorem}
	\label{thm:section1_closeness_lowerbound}
	For any $n$ and even $k$, and small enough $\epsilon$, there exists some $(\epsilon, k)$-wise uniform distribution $\varphi$ over $\{-1, 1\}^n$, such that $\varphi$ is $\delta$-far from every $k$-wise uniform distribution in total variation distance, where 
	\[
	\delta \geq \Omega\left(\frac1k\right)^k n^{k/2} \eps.
	\]
\end{theorem}

Our method for proving this lower bound is again LP duality. Our examples in the lower bound are symmetric distributions with Fourier weight only on level $k$. The density functions then can be written as binary Krawtchouk polynomials which behave similar to Hermite polynomials when $n$ is large. Our dual LP bounds use various properties of Krawtchouk and Hermite polynomials. 

Interestingly both our upper and lower bound utilize LP-duality, which we believe is the most natural way of looking at this problem.

We remark that we can derive a lower bound for odd $k$ from Theorem~\ref{thm:section1_closeness_lowerbound} trivially by replacing $k$ by~$k-1$. There exists a gap of $\sqrt{n}$ between the resulting upper and lower bounds for odd $k$. We believe that the lower bound is tight, and the upper bound may be improvable by a factor of $\sqrt{n}$, as it is in the special case $k = 1$. We leave it as a conjecture for further work:
\begin{conjecture}
	Suppose the distribution $\varphi$ over $\{-1, 1\}^n$ is $(\epsilon, k)$-wise uniform. Then $\varphi$ is $\delta$-close to some $k$-wise uniform distribution in total variation distance, where 
	\[
	\delta \leq O (n^{\lfloor k/2 \rfloor} \eps).
	\]
\end{conjecture}

\subsection{The Testing Problem}

Another application of the Closeness Problem is to property testing of $k$-wise uniformity. Suppose we have sample access from an unknown and arbitrary distribution; we may wonder whether the distribution has a certain property. This question has received tremendous attention in the field of statistics. The main goal in the study of property testing is to design algorithms that use as few samples as possible, and to establish lower bound matching these sample-efficient algorithms. In particular, we consider the property of being~$k$-wise uniform:

\medskip

\emph{Given sample access to an unknown and arbitrary distribution $\varphi$ on $\{-1, 1\}^n$, how many samples do we need to distinguish between the case that $\varphi$ is $k$-wise uniform versus the case that $\varphi$ is $\delta$-far from every $k$-wise uniform distribution?}

\medskip

In this paper, we will refer to this question as the \emph{Testing Problem}. 

We say a testing algorithm is \emph{a $\delta$-tester for $k$-wise uniformity} if the algorithm outputs ``Yes'' with high probability when the distribution $\varphi$ is $k$-wise uniform, and the algorithm outputs ``No'' with high probability when the distribution $\varphi$ is $\delta$-far from any $k$-wise uniform distribution (in total variation distance). 

Property testing is well studied for Boolean functions and distributions. Previous work studied the testing of related properties of distribution, including uniformity \cite{GR11,BFRSW00,RS09} and independence \cite{BFFKRW01,BKR04, ADK15, DK16}. 

The papers \cite{AGM03, AAKMRX07, Xie12} discussed the problem of testing $k$-wise uniformity. \cite{AGM03} constructed a $\delta$-tester for $k$-wise uniformity with sample complexity $O(n^{2k} / \delta^2)$, and \cite{AAKMRX07} improved it to $O(n^k \log^{k+1} n / \delta^2)$. As for lower bounds, \cite{AAKMRX07} showed that $\Omega(n^{(k-1)/2} / \delta)$ samples are necessary, albeit only for $k > 2$.  This lower bound is in particular for distinguishing the uniform distribution from $\delta$-far-from-$k$-wise distributions.

We show a better upper bound for sample complexity:
\begin{theorem}
	\label{thm:section1_testing_upperbound}
	There exists a $\delta$-tester for $k$-wise uniformity of distributions on $\{-1, 1\}^n$ with sample complexity $O\left(\frac1k\right)^{k/2}\frac{n^k}{\delta^2}$. For the special case of $k=1$, the sample complexity is $O\left(\frac{\log n}{\delta^2}\right)$.
\end{theorem}

A natural $\delta$-tester of $k$-wise uniformity is mentioned in \cite{AAKMRX07}: Estimate all Fourier coefficients up to level $k$ from the samples; if they are all smaller than $\epsilon$ then output ``Yes''. In fact this algorithm is exactly attempting to check whether the distribution is $(\eps, k)$-wise uniform. Hence the sample complexity depends on the upper bound for the Closeness Problem. Therefore we can reduce the sample complexity of this algorithm down to $O\left(\frac{n^k \log n}{\delta^2}\right)$ via our improved upper bound for the Closeness Problem. One $\log n$ factor remains because we need to union-bound over the $O(n^k)$ Fourier coefficients up to level $k$. To further get rid of the last $\log n$ factor, we present a new algorithm that estimates the Fourier weight up to level $k$, $\sum_{1 \leq |S| \leq k} \widehat{\varphi}^2(S)$, rather than estimating these Fourier coefficients one by one.

Unfortunately, a lower bound for the Closeness Problem does not imply a lower bound for the Testing Problem directly. In \cite{AAKMRX07}, they showed that a uniform distribution over a random subset of $\{-1, 1\}^n$ of size $O(\frac{n^{k-1}}{\delta^2})$, is almost surely $\delta$-far from any $k$-wise uniform distribution. On the other hand, by the Birthday Paradox, it is hard to distinguish between the fully uniform distribution on all strings of length $n$ and a uniform distribution over a random set of such size. This gives a lower bound for the Testing Problem as $\Omega(n^{(k-1)/2}/\delta)$. Their result only holds for $k > 2$; there was no previous non-trivial lower bound for testing pairwise uniformity. We show a lower bound for the pairwise case. 

\begin{theorem}
	\label{thm:section1_testing_lowerbound}
	Any $\delta$-tester for pairwise uniformity of distributions on $\{-1, 1\}^n$ needs at least $\Omega(\frac{n}{\delta^2})$ samples.
\end{theorem}

For this lower bound we analyze a symmetric distribution with non-zero Fourier coefficients only on level~2. We prove that it is hard to distinguish a randomly shifted version of this distribution from the fully uniform distribution. This lower bound is also better than \cite{AAKMRX07} in that we have a better dependence on the parameter $\delta$ ($\frac{1}{\delta^2}$ rather than $\frac1{\delta}$). Unfortunately we are unable to generalize our lower bound for higher~$k$.

Notice that for our new upper and lower bounds for $k$-wise uniformity testing, there still remains a quadratic gap for $k \geq 2$, indicating that the upper bound might be able to be improved. Both the lower bound in our paper and that in \cite{AAKMRX07} show that it is hard to distinguish between the fully uniform distribution and some specific sets of distributions that are far from $k$-wise uniform. We show that if one wants to improve the lower bound, one will need to use a distribution in the ``Yes'' case that is \emph{not} fully uniform, because we give a sample-efficient algorithm for distinguishing between fully uniform and $\delta$-far from $k$-wise uniform:

\begin{theorem}
	\label{thm:section1_fully_testing}
	For any constant $k$, for testing whether a distribution is fully uniform or $\delta$-far from every $k$-wise uniform distribution, there exists an algorithm with sample complexity $O(k)^k\cdot n^{k/2} \cdot \frac1{\delta^2} \cdot \left(\log \frac{n}{\delta}\right)^{k/2}$.
	
	In fact, for testing whether a distribution is $\alpha k$-wise uniform or $\delta$-far from $k$-wise uniform with $\alpha > 4$, there exists an algorithm with sample complexity  $O(\alpha)^{k/2} \cdot n^{k/2} \cdot \frac1{\delta^2} \cdot \left(\frac{n^k}{\delta^4}\right)^{1/(\alpha - 2)}$. 
\end{theorem}

We remark that testing full uniformity can be treated as a special case of testing $\alpha k$-wise uniformity approximately, by setting $\alpha = \log \frac{n}{\delta}$.

Testing full uniformity has been studied in \cite{GR11,BFRSW00}. Paninski \cite{Paninski08} showed that testing whether an unknown distribution on $\{-1,1\}^n$ is $\Theta(1)$-close to fully uniform requires $2^{n/2}$ samples. Rubinfeld and Servedio \cite{RS09} studied testing whether an unknown monotone distribution is fully uniform or not.

The fully uniform distribution has the nice property that every pair of samples is different in $\frac{n}{2} \pm O(\sqrt{n})$ bits with high probability when the sample size is small. Our algorithm first rejects those distributions that disobey this property. We show that the remaining distributions have small Fourier weight up to level $2k$. Hence by following a similar analysis as the tester in Theorem~\ref{thm:section1_testing_upperbound}, we can get an improved upper bound when these lower Fourier weights are small.

The lower bound remains the same as testing $k$-wise vs.\ far from $k$-wise. Our tester is tight up to a logarithmic factor for the pairwise case, and is tight up to a factor of $\tilde{O}(\sqrt{n})$ when $k > 2$. 

We compare our results and previous best known bounds from \cite{AAKMRX07} in Table~\ref{table}. (We omit constant factors depending on $k$.)

\begin{table}[h]
	
	\centering
	
	\small
	
	\begin{tabular}{ |l||c|c|c|c| }
		\hline
		&
		\multicolumn{2}{c|}{Upper bound} & \multicolumn{2}{c|}{Lower bound} \\
		\cline{2-5}
		& \cite{AAKMRX07} & Our paper & \cite{AAKMRX07} & Our paper\\
		\hline
		\multirow{2}{*}{Closeness Problem} &
		\multirow{2}{*}{$O(n^{k/2} (\log n)^{k/2} \epsilon)$} & {$O(n^{k/2} \epsilon)$ } & \multirow{2}{*}{$\Omega\left(\frac{n^{(k-1)/2}}{\log n} \epsilon\right)$} & \multirow{2}{*}{$\Omega(n^{\lfloor k/2 \rfloor} \epsilon)$} \\
		&&$O(\epsilon)$ for $k=1$&&\\
		\hline
		Testing $k$-wise vs. &\multirow{2}{*}{$\displaystyle O\left(\frac{n^{k} (\log n)^{k + 1} }{\delta^2}\right)$} & {$O\left(\frac{n^{k}}{\delta^2}\right)$ } & \multirow{2}{*}{$\displaystyle \Omega\left(\frac{n^{(k-1)/2}}{\delta}\right)$ for $k > 2$} & \multirow{2}{*}{$\displaystyle \Omega\left(\frac{n}{\delta^2}\right)$ for $k = 2$}\\
		far from $k$-wise&&$O\left(\frac{\log n}{\delta^2}\right)$ for $k = 1$&&\\
		\hline
		Testing $n$-wise vs.\ &\multirow{2}{*}{$\displaystyle O\left(\frac{n^{k} (\log n)^{k + 1} }{\delta^2}\right)$} & {$O\left(\frac{n^{k/2}}{\delta^2} (\log \frac{n}{\delta})^{k/2}\right)$ } & \multirow{2}{*}{$\displaystyle \Omega\left(\frac{n^{(k-1)/2}}{\delta}\right)$ for $k > 2$} & \multirow{2}{*}{$\displaystyle \Omega\left(\frac{n}{\delta^2}\right)$ for $k = 2$}\\
		far from $k$-wise&&$O\left(\frac{\log n}{\delta^2}\right)$ for $k = 1$&&\\		\hline
	\end{tabular}
	\vspace{.2em}	
	\caption{Comparison of our results to \cite{AAKMRX07}}
	\label{table}
\end{table}

\subsection{Organization}
Section~\ref{sec:pre} contains definitions and notations. We will discuss upper and lower bounds for the Closeness Problem in Section~\ref{sec:closeness}. We will discuss the sample complexity of testing $k$-wise uniformity in Section~\ref{sec:testing}. We present a tester for distinguishing between $\alpha k$-wise uniformity (or fully uniformity) and far-from $k$-wise uniformity in Section~\ref{sec:fully}.

\section{Preliminaries}
\label{sec:pre}

\subsection{Fourier analysis of Boolean functions}

We use $[n]$ to denote the set $\{1, \dots, n\}$. We denote the symmetric difference of two sets $S$ and $T$ by $S \oplus T$. For Fourier analysis we use notations consistent with \cite{OD14}. Every function $f : \{-1, 1\}^n \to \R$ has a unique representation as a multilinear polynomial
\[
f(x) = \sum_{S \subseteq [n]} \widehat{f}(S) x^S \quad \text{where} \quad x^S = \prod_{i \in S} x_i.
\]
We call $\widehat{f}(S)$ the Fourier coefficient of $f$ on $S$. We use $\bx \sim \{-1, 1\}^n$ to denote that $\bx$ is uniformly distributed on~$\{-1, 1\}^n$. We can represent Fourier coefficients as
\[
\widehat{f}(S) = \E_{\bx \sim \{-1, 1\}^n} \left[f(\bx) \bx^S\right].
\]

We define an inner product $\langle \cdot, \cdot \rangle$ on pairs of functions $f, g : \{-1,1\}^n \to \R$ by
\[
\langle f, g \rangle = \E_{\bx \sim \{-1, 1\}^n} [f(\bx) g(\bx)] = \sum_{S \subseteq[n]}\widehat{f}(S) \widehat{g}(S).
\] 

We introduce the following $p$-norm notation: $\|f\|_p = \left(\E[|f(\bx)|^p]\right)^{1/p}$,
and the Fourier $\ell_p$-norm is
$\|\widehat{f}\|_p = \left(\sum_{S \subseteq [n]}|\widehat{f}(S)|^p\right)^{1/p}$.

We say the \emph{degree} of a Boolean function, $\textnormal{deg}(f)$ is $k$ if its Fourier polynomial is degree $k$. We denote $f^{=k}(x) =  \sum_{|S| = k} \widehat{f}(S) x^S$, and $f^{\leq k}(x) =  \sum_{|S| \leq k} \widehat{f}(S) x^S$. We denote the \emph{Fourier weight} on level $k$ by~$\W{k}[f] = \sum_{|S| = k} \widehat{f}(S)^2$. We denote~$\W{1\dots k}[\varphi] = \sum_{1 \leq |S| \leq k} \wh{\varphi}(S)^2$.

We define the convolution $f * g$ of a pair of functions $f, g : \{-1, 1\}^n \to \R$ to be
\[
(f * g)(x) = \E_{\by \sim \{-1, 1\}^n}[f(x)g(x \circ \by)],
\]
where $\circ$ denotes entry-wise multiplication. The effect of convolution on Fourier coefficients is that $\widehat{f * g}(S) = \wh{f}(S) \wh{g}(S)$.

\subsection{Densities and distances}

When working with probability distribution on $\{-1, 1\}^n$, we prefer to define them via \emph{density} function. A density function $\varphi : \{-1, 1\}^n \to \R^{\geq 0}$ is a nonnegative function satisfying
$
\wh{\varphi}(\emptyset) = \E_{\bx \sim \{-1, 1\}^n}[\varphi(\bx)] = 1.
$
We write $\by \sim \varphi$ to denote that $\by$ is a random variable drawn from the distribution $\varphi$, defined by
\[
\Pr_{\by \sim \varphi}[\by = y] = \frac{\varphi(y)}{2^n},
\]
for all $y \in \{-1, 1\}^n$. We identify distributions with their density functions when there is no risk of confusion.

We denote $\varphi^{+t}(x) = \varphi(x \circ t)$. We denote by $\bone_A$ the density function for the uniform distribution on support set $A$. The density function associated to the fully uniform distribution is the constant function $\bone$.

The following lemma about density functions of degree at most $k$ derives from Fourier analysis and hypercontractivity.

\begin{lemma}
	\label{lem:fourierl1bound}
	Let $\varphi : \{-1, 1\}^n \to \R^{\geq 0}$ be a density function of degree at most $k$. Then 
	\[
	\|\widehat{\varphi}\|_2 = \sqrt{\sum_{S} \widehat{\varphi}(S)^2} \leq e^k.
	\]
\end{lemma}
\begin{proof}
	\[
	\|\widehat{\varphi}\|_2 = \|\varphi\|_2 \leq e^k \|\varphi\|_1 = e^k.
	\]
	The first equality holds by Parseval's Theorem (see Section~1.4 in \cite{OD14}). The inequality holds by hypercontractivity (see Theorem~9.22 in \cite{OD14}). The last equality holds since $\varphi$ is a density function.
\end{proof}

A distribution $\varphi$ over $\{-1, 1\}^n$ is \emph{$k$-wise uniform} if and only if $\widehat{\varphi}(S) = 0$ for all $1 \leq |S| \leq k$ (see Chapter 6.1 in \cite{OD14}). We say that distribution $\varphi$ over $\{-1, 1\}^n$ is \emph{$(\epsilon, k)$-wise uniform} if $|\widehat{\varphi}(S)| \leq \eps$ for all~$1 \leq S \leq k$.

The most common way to measure the distance between two probability distributions is via their \emph{total variation distance}.  If the distributions have densities $\varphi$ and $\psi$, then the total variation distance is defined to be 
\[
d_{\textnormal{TV}}(\varphi, \psi) = \sup_{A \subseteq \{-1, 1\}^n} \left|\Pr_{\bx \sim \varphi}[\bx \in A] - \Pr_{\bx \sim \psi}[\bx \in A]\right| = \frac12\E_{\bx}\left[|\varphi(\bx) - \psi(\bx)|\right] =  \frac12 \|\varphi - \psi \|_1.
\]
We say that $\varphi$ and $\psi$ are \emph{$\delta$-close} if $d_{\textnormal{TV}}(\varphi, \psi) \leq \delta$.

Supposing $H$ is a set of distributions, we denote 
\[
d_{\textnormal{TV}}(\varphi, H) = \min_{\psi \in H} d_{\textnormal{TV}}(\varphi, \psi).
\]
In particular, we denote the set of $k$-wise uniform densities by $\kWISE{k}$. We say that density $\varphi$ is \emph{$\delta$-close to $k$-wise uniform} if $d_{\textnormal{TV}}(\varphi, \kWISE{k}) \leq \delta$, and is \emph{$\delta$-far} otherwise.

\subsection{Krawtchouk and Hermite polynomials}

Krawtchouk polynomials were introduced in \cite{Krawtchouk1929}, and arise in the analysis of Boolean functions as shown in \cite{Levenshtein95,Kalai2002}.  Consider the following Boolean function of degree $k$ and input length $n$: $f(x) = \sum_{|S| = k} x^S$. It is symmetric and therefore only depends on the Hamming weight of $x$. Let $t$ be the number of $-1$'s in $x$. Then the output of $f$ is exactly the same as the Krawtchouk polynomial $K_k^{(n)}(t)$.
\begin{definition}
	We denote by $K_k^{(n)}(t)$ the Krawtchouk polynomial:
	\[
	K_k^{(n)}(t) = \sum_{j = 0}^{k}(-1)^j {t \choose j}{n-t \choose k-j},
	\]
	for $k = 0, 1, \dots, n$.
	
\end{definition}

We will also use Hermite polynomials in our analysis.
\begin{definition}
	We denote by $h_k(z)$ the normalized Hermite polynomial:
	\[
	h_k(z) = \frac1{\sqrt{k!}} (-1)^k e^{\frac12z^2} \frac{d^k}{dz^k} e^{-\frac12z^2}.
	\]
	Its explicit formula is
	\[
	h_k(z) = \sqrt{k!} \cdot \left(\frac{z^k}{0!! \cdot k!} - \frac{z^{k-2}}{2!! \cdot (k-2)!} + \frac{z^{k-4}}{4!! \cdot (k-4)!} - \frac{z^{k-6}}{6!! \cdot (k-6)!} + \cdots\right).
	\]
\end{definition}

One useful fact is that the derivative of a Hermite polynomial is a scalar multiple of a Hermite polynomial (see Exercise~11.10 in \cite{OD14}):

\begin{fact}
	\label{fact:hermite}
	For any integer $k \geq 1$, we have
	\[
	\frac{d}{dz}h_k(z) = \sqrt{k} h_{k-1}(z).
	\] 
\end{fact}

The relationship between Krawtchouk and Hermite polynomials is that we can treat Hermite polynomials as a limit version of Krawtchouk polynomials when $n$ goes to infinity (see Exercise~11.14 in \cite{OD14}).

\begin{fact}
	\label{fact:kravchuk}
	For all $k \in \N$ and $z \in \R$ we have
	\[
	{n \choose k}^{-1/2} \cdot K_k^{(n)}\left(\frac{n-z\sqrt{n}}2\right) \xrightarrow{n \to \infty} h_k(z).
	\]
\end{fact}

Instead of analyzing Krawtchouk polynomials, it is easier to study Hermite polynomials when $n$ is large because Hermite polynomials have a more explicit form. We present some basic properties of Hermite polynomials with brief proofs.

\begin{lemma}
	\label{lem:hermite}
	The following are properties of $h_k(z)$:
	\begin{enumerate}
		\item $|h_k(z)| \leq h_k(k)$ for any $|z| \leq k$;
		\item $h_k(z)$ is positive and increasing when $z \geq k$;
		\item $h_k(Ck) \leq (Ck)^k/\sqrt{k!}$ for any constant $C \geq 1$.
	\end{enumerate}
\end{lemma}

\begin{proof}
	We will treat the case of $k = 4i + 2$ for some integer $i$. The proof for  the general case is similar. When $k =  4i + 2$, we can group adjacent terms into pairs:
	\[
	h_k(z) = \sqrt{k!} \cdot \sum_{i = 0}^{(k-2)/4} \frac{z^{k-4i-2}}{(4i+2)!! \cdot (k-4i)!}((4i+2)z^2 - (k-4i)(k-4i - 1)).
	\]
	\begin{enumerate}
		\item Notice that $|(4i+2)z^2 - (k-4i)(k-4i - 1)|$ is always between $-(k-4i)(k-4i - 1)$ and $(4i+2)k^2 - (k-4i)(k-4i - 1)$ when $|z| \leq k$. Both the upper and lower bound have absolute value at most $(4i+2)k^2 - (k-4i)(k-4i - 1)$. Therefore by the triangle inequality we have $|h_k(z)| \leq h_k(k)$.
		\item It is easy to check that $((4i+2)z^2 - (k-4i)(k-4i - 1))$ is positive when $z \geq k$. Then by Fact~\ref{fact:hermite}, $\frac{d}{dz}h_k(z) = \sqrt{k} h_{k - 1}(z) > 0$ when $z \geq k$.
		\item This is trivial from the explicit formula since each term is exactly smaller than the previous term when $z \geq k$. \qedhere
	\end{enumerate}
\end{proof}

\section{The Closeness Problem}
\label{sec:closeness}

In this section, we prove the upper bound in Theorem~\ref{thm:section1_closeness_upperbound} and the lower bound in Theorem~\ref{thm:section1_closeness_lowerbound}. One interesting fact is that we use duality of linear programming (LP) in both the upper and lower bound. We think this is the proper perspective for analyzing these questions.

\subsection{Upper bound}

The key idea for proving the upper bound is mixture distributions. Given an ($\epsilon$, $k$)-wise uniform density $\varphi$, we try to mix it with some other distribution $\psi$ using mixture weight $w$, such that the mixture distribution $\frac1{1+w}(\varphi + w \psi)$ is $k$-wise uniform and is close to the original distribution. The following lemma shows that the distance between the original distribution and the mixture distribution is bounded by the weight $w$.

\begin{lemma}
	\label{lem:weighted}
	If $\varphi' =\frac1{1+w}(\varphi + w \psi)$ for some $0 \leq w \leq 1$ and density functions $\varphi, \psi$, then $d_{\textnormal{TV}}(\varphi, \varphi')\leq w$.
\end{lemma}
\begin{proof}
	$
	d_{\textnormal{TV}}(\varphi, \varphi') = \frac12 \|\varphi'-\varphi\|_1 = \frac 12 \|\varphi' - ((1+w)\varphi' - w\psi))\|_1 = \frac12 w\|\varphi'-\psi\|_1 \leq w.
	$
\end{proof}

Therefore we only need to show the existence of an appropriate $\psi$ for some small $w$. The constraints on~$\psi$ can be written as an LP feasibility problem. Therefore by Farkas' Lemma we only need to show that its dual is not feasible. The variables in the dual LP can be seen as a density function of degree at most $k$.

\begin{proof}[Proof of Theorem~\ref{thm:section1_closeness_upperbound} (general $k$ case)]
	Given density function $\varphi$, we try to find another density function $\psi$ with constraints
	\[
	\widehat{\psi}(S) = -\frac1w \widehat{\varphi}(S)
	\]
	for all $1 \leq |S| \leq k$. Suppose such a density function $\psi$ exists. Then it is trivial that $\frac{\varphi + w\psi}{1 + w}$ is also a density function and is $k$-wise uniform. By Lemma~\ref{lem:weighted}, we conclude that $d_{\textnormal{TV}}(\varphi, \kWISE{k}) \leq w$.
	
	The rest of proof is to show that such a $\psi$ exists when $w = e^k\sqrt{\W{1\dots k}[\varphi]}$. We can write the existence as an LP feasibility problem with variables $\psi(x)$ for $x \in \{-1, 1\}^n$ and constraints:
	\begin{align*}
	\wh{\psi}(\emptyset) &= 1, & \\
	\wh{\psi}(S) &= -\frac1w \wh{\varphi}(S), & \forall 1 \leq |S| \leq k,\\
	\psi(x) &\geq 0, & \forall x \in \{-1, 1\}^n,
	\end{align*}
	where $\wh{\psi}(S) = \E[\psi(\bx)\bx^S]$ is a linear combination of variables $\psi(x)$.
	
	The dual LP has variables $\psi'(x)$ for $x \in \{-1, 1\}^n$ with constraints:
	\begin{align*}
	\wh{\psi'}(\emptyset) &= 1, & \\
	\wh{\psi'}(S) &= 0, & \forall |S| > k,\\
	\psi'(x) &\geq 0, &\forall x \in \{-1, 1\}^n,\\
	\frac1w \sum_{1 \leq |S| \leq k} \wh{\varphi}(S) \wh{\psi'}(S)  &> 1.
	\end{align*}
	The original LP is feasible if and only if its dual LP is infeasible, by Farkas' Lemma. This completes the proof, since when $w = e^k \sqrt{\W{1\dots k}[\varphi]}$, for any density function $\psi'$ with degree $k$ we have 
	\[
	\frac1w \sum_{1 \leq |S| \leq k} \wh{\varphi}(S) \wh{\psi'}(S) \leq \frac1{e^k\sqrt{\W{1\dots k}[\varphi]}} \sum_{1 \leq |S| \leq k} |\wh{\varphi}(S)| |\wh{\psi'}(S)| \leq \frac1{e^k} \|\wh{\psi'}\|_2 \leq 1,
	\]
	where the second inequality holds by Cauchy--Schwarz, and the last inequality holds by Lemma~\ref{lem:fourierl1bound} since $\psi'$ has degree at most $k$.
\end{proof}

For $k = 1$, further improvement can be achieved. We still try to use mixture distributions. Here we want to mix the distribution $\varphi$ with indicator distributions on subsets of coordinates that have opposite biases to those of the original distribution.

\begin{proof}[Proof of Theorem~\ref{thm:section1_closeness_upperbound} (case $k = 1$)]
	By identifying each $x_i$ with $-x_i$ if necessary, we may assume without loss of generality that $\wh{\varphi}(\{i\}) \geq 0$ for all $i$.  In addition, by reordering the coordinates, we may assume without loss of generality that $0 \leq \wh{\varphi}(\{1\}) \leq \dots \leq \wh{\varphi}(\{n\}) = \eps$. Define $\psi_j$ to be the density of the distribution over $\{-1, 1\}^n$ which is uniform on coordinates $x_1, \dots,  x_{j-1}$, and has $x_i$ constantly fixed to be $-1$ for $j \leq i \leq n$. It is easy to check $\wh{\psi_j}(\{i\}) = 0$ for $i < j$ and $\wh{\psi_j}(\{i\}) = -1$ for $i \geq j$.
	
	We define $\varphi'$ as
	\[
	\varphi' = \frac1{1+\eps}\left(\varphi + \sum_{j=1}^n w_j \psi_j\right),
	\]
	where
	\[
	w_1 = \wh{\varphi}(\{1\}), \qquad w_j =  \wh{\varphi}(\{j\}) - \wh{\varphi}(\{{j-1}\}) \quad \forall 1 < j \leq n.
	\]
	
	It is easy to check that $\varphi'$ is a density function and 
	\[
	\wh{\varphi'}(\{i\}) = \frac1{1+\eps} \left(\wh{\varphi}(\{i\}) + \left( \sum_{j=1}^i w_j \right) (-1) \right) = 0.
	\]
	Therefore $\varphi'$ is 1-wise uniform.
	Then by Lemma~\ref{lem:weighted},
	\[
	d_{TV}(\varphi, \kWISE{1}) \leq \frac12 \|\varphi - \varphi'\|_1 \leq \sum_{j=1}^n w_j = \eps. \qedhere
	\]
\end{proof}

\subsection{Lower bound}

Interestingly, our proof of the lower bound also utilizes LP duality. We can write the Closeness Problem in the form of linear programming with variables $\varphi'(x)$ for $x \in \{-1, 1\}^n$, as follows: 
\begin{align*}
&\text{minimize} & d_{\text{TV}}(\varphi, \varphi') &= \frac12 \|\varphi - \varphi'\|_1 &\\
&\text{subject to:} & \wh{\varphi'}(\emptyset) &= 1, & \\
&& \wh{\varphi'}(S) &= 0, & \forall 1 \leq |S| \leq k,\\
&& \varphi'(x) &\geq 0, & \forall x \in \{-1, 1\}^n.
\end{align*}

We ignore the factor of $1/2$ in the minimization for convenience in the following analysis.

The dual LP, which has variables $p(x), q(x)$ for $x \in \{-1, 1\}^n$, is the following:
\begin{align*}
&\text{maximize} & \langle \varphi, q\rangle - \wh{p}(\emptyset) & &\\
&\text{subject to:} & p(x) - q(x) &\geq 0, & \forall x \in \{-1, 1\}^n,\\
&& q(x) & \leq 1, &\forall x \in \{-1, 1\}^n,\\
&& p(x) & \geq -1, &\forall x \in \{-1, 1\}^n,\\
&& \textnormal{deg}(p) &\leq k. &
\end{align*}

Thus given a pair of Boolean functions $p, q$ satisfying the constraints, the quantity $\langle \varphi, q\rangle - \wh{p}(\emptyset)$ is a lower bound for our Closeness Problem. Our distribution $\varphi$ achieving the lower bound is a symmetric polynomial, homogeneous of degree $k$ (except that it has a constant term of $1$, as is necessary for every density function). We can use Krawtchouk and Hermite polynomials to simplify the analysis.

\begin{proof}[Proof of Theorem~\ref{thm:section1_closeness_lowerbound}]
	We define 
	\[
	\varphi(x) = 1 + \mu {n \choose k }^{-1/2} \sum_{|S| = k} x^S, \quad p(x) = \mu {n \choose k }^{-1/2} \sum_{|S| = k} x^S, \quad q(x) = \min(p(x), 1),
	\]
	where $\mu$ is a small parameter to be chosen later that will ensure $\varphi(x) \geq 0$ and $p(x) \geq -1$ for all $x \in \{-1, 1\}^n$. We have $\eps = \max_{1 \leq |S| \leq k}|\wh{\varphi}(S)| = \mu {n \choose k }^{-1/2}$. 
	
	Since $\wh{p}(\emptyset) = 0$, the objective function of the dual LP is
	\begin{align*}
	\langle \varphi, q \rangle &= \langle \varphi, \min(p, 1) \rangle = \langle \varphi, 1_{p > 1} \rangle + \langle \varphi, p1_{p\leq 1} \rangle = \langle \varphi, p \rangle - \langle \varphi, (p-1)1_{p > 1} \rangle \\
	&\geq \langle \varphi, p \rangle - \sqrt{\Pr_{\bx \sim \varphi}[p(\bx) > 1] \cdot \langle\varphi, (p-1)^2\rangle},
	\end{align*}
	where the last inequality holds by Cauchy--Schwarz.
	It is easy to calculate the inner products $\langle \varphi, p \rangle = \mu^2 $, and
	\begin{align*}
	\langle \varphi, (p-1)^2 \rangle &= \langle \varphi, p^2 \rangle - 2\langle \varphi, p \rangle + 1 \\
	&= \mu^2 + \mu^3 {n \choose k}^{-1/2} {k \choose k/2} {n-k \choose k/2} -2\mu^2 + 1 \\
	&\leq 1 + \mu^3 {k \choose k/2}^{3/2} - \mu^2. 
	\end{align*}
	Assuming $\mu < 2^{-\frac32k}$, we have $\langle \varphi, (p-1)^2 \rangle < 1$.
	
	Now we need to upper bound $\Pr_{\bx \sim \varphi}[p(\bx) > 1]$. Define $\bz$ satisfying  $(n -\bz\sqrt{n})/2 = \sum_i \bx_i$. Then
	\[
	\Pr_{\bx \sim \varphi}[p(\bx) > 1] = \Pr_{\bx \sim \varphi}\left[\mu {n \choose k}^{-1/2} \cdot K_k\left(\frac{n - \bz \sqrt{n}}2, n\right) > 1 \right].
	\]
	
	By Fact~\ref{fact:kravchuk}, we know that when $z \leq k$, for sufficient large $n$, 
	\[
	{n \choose k}^{-1/2} \cdot K_k\left(\frac{n-z\sqrt{n}}2,n\right) < 2h_k(z).
	\]
	
	Now we set $\mu = \frac{\sqrt{k!}}{2(Ck)^k}$ with some constant $C \geq 1$. It is easy to check that $\mu < 2^{-\frac32k}$. Using the properties in Lemma~\ref{lem:hermite}, we get
	\begin{align*}
	\Pr_{\bx \sim \varphi}\left[\mu {n \choose k}^{-1/2} \cdot K_k\left(\frac{n - \bz \sqrt{n}}2, n\right) > 1 \right] &\leq \Pr_{\bx \sim \varphi}[2\mu h_k(\bz)> 1] \\
	&\leq \Pr_{\bx \sim \varphi}[h_k(\bz) > h_k(Ck)] \\
	&= \Pr_{\bx \sim \varphi}[|\bz| > Ck].
	\end{align*}
	
	Then using Cauchy--Schwarz again, we get
	\begin{align*}
	\Pr_{\bx \sim \varphi}[|\bz| > Ck] 
	&\leq \sqrt{\E_{\bx \sim \{-1, 1\}^n}[\varphi(\bx)^2]} \sqrt{\Pr_{\bx \sim \{-1, 1\}^n}[|\bz| > Ck]} \\
	&\leq \sqrt{1+ \mu^2} \sqrt{2\exp(-C^2k^2/2)} \\
	&\leq 2\exp(-(Ck)^2/4).
	\end{align*}
	
	Therefore we get that the objective function is at least
	\[
	\langle \varphi, p \rangle - \sqrt{\Pr_{\bx \sim \varphi}[p(\bx) > 1] \cdot \langle\varphi, (p-1)^2\rangle} \geq \mu^2 -\sqrt{2\exp(-(Ck)^2/4)} \geq \Omega\left(\frac1k\right)^k.
	\]
	The last inequality holds when we choose a sufficiently large constant $C$.
	
	This completes the proof, because $\varphi$ is at least $\delta$-far from $k$-wise uniform with $\delta =  \Omega\left(\frac1k\right)^k$, and we have $\eps = \mu {n \choose k }^{-1/2} \leq \frac{n^{-k/2}}{2^{\Omega(k)}}$. Therefore we have $\delta \geq \Omega\left(\frac1k\right)^k n^{k/2} \eps$.
\end{proof}

\section{The Testing Problem}

\label{sec:testing}
In this section, we study the problem of testing whether a distribution is $k$-wise uniform or $\delta$-far from $k$-wise uniform. These bounds are based on new bounds for the Closeness Problem. We present a new testing algorithm for general $k$ in Section~\ref{sec:testing_upper_new}. We give a lower bound for the pairwise case in Section~\ref{sec:testing_lower}.

\subsection{Upper bound}

\label{sec:testing_upper_new}

Given $m$ samples from $\varphi$, call them $\bx_1, \dots, \bx_m$, we will first show that
\[
\Delta (\bX) = \avg_{1 \leq s < t \leq m} \left(\sum_{1 \leq |S|\leq k} \bx_s^S\bx_t^S\right)
\]
is a natural estimator of $\W{1 \dots k}[\varphi]$.

\begin{lemma}
	\label{lem:mean-and-var}
	It holds that
	\begin{align}
	\mu = \E[\Delta(\bX)] &= \W{1\dots k}[\varphi]; \nonumber\\
	\label{eq:var}
	\Var[\Delta(\bX)] &\leq \frac{4}{m^2}L_k(\varphi)+ \frac{4}{m}\sqrt{L_k(\varphi)} \mu,
	\end{align}
	where $L_k(\varphi) = \sum_{1 \leq |S_1|, |S_2| \leq k} \wh{\varphi}(S_1 \oplus S_2)^2$.
\end{lemma}

\begin{proof}
	
	We denote $F(x, y) = \sum_{1 \leq |S| \leq k} x^S y^S$. We know that
	\[ 
	\E_{\bx, \by \sim \varphi}[\bx^S \by^S] = \E_{\bx \sim \varphi}[\bx^S]\E_{\by \sim \varphi}[\by^S] = \wh{\varphi}(S)^2,
	\]
	when $\bx$ and $\by$ are independent samples drawn from $\varphi$. Therefore by linearity of expectation, $\E_{\bx, \by \sim \varphi}[F(\bx, \by)] = \W{1\dots k}[\varphi]$, and clearly by taking the average, 
	\[
	\mu = \E[\Delta(\bX)] = \E[\text{avg}_{s < t} F(\bx_s, \bx_t)] = \text{avg}_{s < t} \E[F(\bx_s, \bx_t)] = \W{1\dots k}[\varphi].
	\]
	
	We need to expand the variance:
	\begin{equation}
	\label{eq:expand_variance_delta}
	\Var\left[\avg_{s < t} (F(\bx_s, \bx_t))\right]
	= \frac1{{{m \choose 2}}^2}\sum_{\substack{s < t\\ s' < t'}} \Cov[F(\bx_s, \bx_t), F(\bx_{s'}, \bx_{t'})].
	\end{equation}
	We will discuss these covariances in three cases.
	
	\textbf{Case 1:} $| \{s, t\} \cap \{s', t'\}| = 2$. Let $\bx, \by \sim \varphi$ be independent random variables.
	\[
	\Cov[F(\bx, \by), F(\bx, \by)] = \Var_{\bx, \by \sim \varphi}[F(\bx, \by)] \leq \E_{\bx, \by \sim \varphi}[F(\bx, \by)^2] = \E_{\bx, \by \sim \varphi}\left[\left(\sum_{1\leq |S| \leq k}\bx^S \by^S\right)^2\right].
	\]
	Notice here all $\bx_i$'s are Rademacher variables with $\bx_i^2 = 1$, and similarly for the $\by_i$'s. Therefore
	\begin{align*}
	\E_{\bx, \by \sim \varphi}\left[\left(\sum_{1\leq |S| \leq k}\bx^S \by^S\right)^2\right] 
	&= \sum_{1\leq |S_1|, |S_2| \leq k} \E_{\bx, \by \sim \varphi}\left[\bx^{S_1 \oplus S_2} \by^{S_1 \oplus S_2}\right] \\
	&= \sum_{1\leq |S_1|, |S_2| \leq k} \wh{\varphi}(S_1 \oplus S_2)^2 = L_k(\varphi).
	\end{align*}
	
	\textbf{Case 2:} $| \{s, t\} \cap \{s', t'\}| = 1$. Let $\bx, \by, \bz \sim \varphi$ be independent random variables. Similar to Case 1, we have:
	\begin{align*}
	\Cov[F(\bx, \by), F(\bx, \bz)] &\leq \E[F(\bx, \by)F(\bx, \bz)] \\
	&=\E\left[\left(\sum_{1\leq |S_1| \leq k}\bx^{S_1}\by^{S_1}\right)\left(\sum_{1\leq |S_2| \leq k}\bx^{S_2}\bz^{S_2}\right)\right]\\
	&=\E\left[\sum_{1\leq |S_1|, |S_2| \leq k} \bx^{S_1 \oplus S_2} \by^{S_1} \bz^{S_2}\right]\\
	& = \sum_{1\leq |S_1|, |S_2| \leq k} \wh{\varphi}(S_1 \oplus S_2) \wh{\varphi}(S_1) \wh{\varphi}(S_2) \\
	& \leq \sqrt{\sum_{1 \leq |S_1|, |S_2| \leq k} \wh{\varphi}(S_1 \oplus S_2)^2} \sqrt{\sum_{1 \leq |S_1|, |S_2| \leq k}\wh{\varphi}(S_1)^2 \wh{\varphi}(S_2)^2}\\
	& = \sqrt{L_k(\varphi)}\mu,
	\end{align*}
	where the inequality comes from Cauchy--Schwarz.
	
	\textbf{Case 3:} $| \{s, t\} \cap \{s', t'\}| = 0$. Let $\bx, \by, \bz, \bw \sim \varphi$ be independent random variables. Clearly $F(\bx, \by)$ and $F(\bz, \bw)$ are independent and therefore $\Cov[F(\bx, \by), F(\bz, \bw)] = 0$.
	
	Plugging all these cases into \cref{eq:expand_variance_delta}, we get
	\begin{align*}
	\Var[\Delta(\bX)] &= \Var\left[\avg_{s < t} (F(\bx_s, \bx_t)\right]\\
	&= \frac1{{{m \choose 2}}^2}\left({m \choose 2} L_k(\varphi) + m(m - 1)(m - 2) \sqrt{L_k(\varphi)} \mu\right)\\
	&\leq \frac{4}{m^2} L_k(\varphi) + \frac4{m} \sqrt{L_k(\varphi)} \mu. \qedhere
	\end{align*}
	
\end{proof}

Given Lemma~\ref{lem:mean-and-var} we can bound the samples we need for estimating $\W{1\dots k}[\varphi]$.

\begin{theorem}[$\W{1\dots k}$ Estimation Test]
	\label{thm:sample}
	Let $\varphi :\{-1, 1\}^n \to \R^{\geq 0}$ be a density function, promised to satisfy $\W{i}[\varphi] \leq A n^{i/2}$ for all $i = 0, 1, \dots, 2k$. There is an algorithm that, given 
	\begin{equation}
	\label{eq:samplebound}
	m \geq 1000\frac{2^{k} \sqrt{A} n^{k/2}}{\theta}
	\end{equation}
	samples, distinguishes with probability at least $3/4$ whether $\W{1\dots k}[\varphi] \leq \frac12 \theta$ or $\W{1\dots k}[\varphi] > \theta$.
\end{theorem}

\begin{proof}
	The algorithm is simple: we report ``$ \mu \leq \frac12 \theta$'' if $\Delta(\bX) \leq \frac34 \theta$ and report ``$\mu > \theta$'' if $\Delta(\bX) > \frac34 \theta$. 
	
	Now we need to bound $L_k(\varphi)$ to bound the variance of $\Delta(\bX)$. For a fixed subset $|S| \leq 2k$, how many pairs of $1 \leq |S_1|, |S_2| \leq k$ are there satisfying $S = S_1 \oplus S_2$? We denote $S_1 = S_1' \cup T$, $S_2 = S_2' \cup T$, where~$S_1', S_2', T$ are disjoint. Then $S = S_1' \cup S_2'$. For a fixed set $S$, there are at most $2^{|S|}$ different ways to split it into two sets $S_1', S_2'$. Because $\max\{S_1', S_2'\} \geq \lceil|S|/2\rceil$ and $|S_1|, |S_2| \leq k$, we have $|T| \leq k - \lceil|S|/2 \rceil$. Therefore there are at most 
	\[
	\sum_{j = 0} ^ {k - \lceil |S|/2 \rceil}{n - |S| \choose j} \leq \frac{2n^{k - \lceil |S|/2 \rceil}}{(k - \lceil |S|/2 \rceil)!}
	\]
	ways to choose the set $T$ for any fixed $S_1', S_2'$. Hence,
	\begin{align*}
	L_k(\varphi) &= \sum_{1\leq |S_1|, |S_2| \leq k} \wh{\varphi}(S_1 \oplus S_2)^2 \\
	&= \sum_{|S| \leq 2k} \sum_{\substack{S_1' \cap S_2' = \emptyset \\ S_1' \cup S_2' = S}} \quad
	\sum_{\substack{T \cap S_1' = \emptyset, T \cap S_2' = \emptyset \\ |T| + \max\{|S_1'|, |S_2'|\} \leq k}} \wh{\varphi} (S)^2 \\
	& \leq \sum_{|S| \leq 2k} 2^{|S|} \frac{2n^{k - \lceil |S|/2 \rceil}}{(k - \lceil |S|/2 \rceil)!} \wh{\varphi} (S)^2 \\
	& = \sum_{i = 0}^{2k} 2^i \frac{2n^{k - \lceil i/2 \rceil}}{(k - \lceil i/2 \rceil)!} \W{i}[\varphi].
	\end{align*}
	Plugging in $\W{i}[\varphi] \leq An^{i/2}$, we get
	\begin{equation}
	\label{eq:boundlk}
	L_k(\varphi) \leq \sum_{i = 0}^{2k} 2^i \frac{2n^{k - \lceil i/2 \rceil}}{(k - \lceil i/2 \rceil)!} \W{i}[\varphi] \leq 2^{2k+2} An^k.
	\end{equation}
	
	By substituting \cref{eq:boundlk} and \cref{eq:samplebound} into \cref{eq:var}, we have
	\[
	\Var[\Delta(\bX)] \leq \frac{4}{500^2}\theta^2 + \frac{4}{500} \theta \mu \leq \frac1{64} \max \{\theta^2, \mu^2\}.
	\]
	Then we conclude our proof by Chebyshev's inequality:
	\begin{align*}
	\Pr\left[|\Delta(\bX) - \mu| \leq \frac14 \max\{\theta, \mu\}\right] 
	&\geq \Pr\left[|\Delta(\bX) - \mu| \leq 2 \sqrt{\Var[\Delta(\bX)]}\right]\\
	&\geq 1 - \left(\frac12\right)^2 = \frac34. \qedhere
	\end{align*}
\end{proof}

This $\W{1\dots k}$ Estimation Test is just what we need for testing $k$-wise uniformity with the upper bound of the Closeness Problem.

\begin{proof}[Proof of Theorem~\ref{thm:section1_testing_upperbound}]
	From Theorem~\ref{thm:section1_closeness_upperbound} we know that if density $\varphi$ is $\delta$-far from $k$-wise uniform, then $\W{1 \dots k}[\varphi] > \left(\frac{\delta}{e^k}\right)^2$. On the other hand if $\varphi$ is $k$-wise uniform, by definition we have $\W{1 \dots k}[\varphi] = 0$. Therefore distinguishing between $k$-wise uniform and $\delta$-far from $k$-wise uniform can be reduced to distinguishing between $\W{1 \dots k}[\varphi] > \left(\frac{\delta}{e^k}\right)^2$ and $\W{1 \dots k}[\varphi] = 0$. 
	
	For any density function $\varphi$, $|\wh{\varphi}(S)| = \left| \E[\varphi(\bx) \bx^S] \right| \leq 1$ for any $S \subseteq [n]$. Therefore assigning $A = n^k$, we have 
	\[
	\W{i}[\varphi] = \sum_{|S| = i} \wh{\varphi}(S)^2 \leq n^i \leq A n^{i/2}
	\] 
	for every $i = 0, 1, \dots, 2k$.
	
	Hence we can run the $\W{1\dots k}$ Estimator Test in Theorem~\ref{thm:sample} with parameter $\theta =  \left(\frac{\delta}{e^k}\right)^2$ and $A = n^k$, thereby we solve the Testing Problem with sample complexity $2^{O(k)}n^k/\delta^2$.
	
	In fact by mroe precise calculation we can further improve the constant factor involving $k$ to $O\left(\frac1k\right)^{k/2}$, but we will omit the proof here for the sake of brevity.
\end{proof}

\subsection{Lower bound for the pairwise case}
\label{sec:testing_lower}

An upper bound for the Closeness Problem implies an upper bound for the Testing Problem. But a lower bound for Closeness does not obviously yield a lower bound for the Testing Problem. The function used to show the lower bound for the Closeness Problem is far from $k$-wise uniform, but it is not sufficient to say that it is hard to distinguish between it and some $k$-wise uniform distribution. In \cite{AAKMRX07}, they show that it is hard to distinguish between the fully uniform distribution and the uniform distribution on a random set of size around $O(n^{k-1}/\delta^2)$; this latter distribution is far from $k$-wise uniform with high probability for $k > 2$. 

We show that the density function $\varphi$ we used for the lower bound for the Closeness Problem is a useful density to use for a testing lower bound in the pairwise case. However it is not hard to distinguish between the fully uniform distribution and $\varphi$. Our trick is shifting $\varphi$ by a random ``center''. We remind the reader that we denote by $\varphi^{+t}(x) = \varphi (x \circ t)$ the distribution $\varphi$ shifted by vector $t$. We claim that with $m = o(n/\delta^2)$ samples, it is hard to distinguish the fully uniform distribution from $\varphi^{+t}$ with a uniformly randomly chosen~$t$.

\begin{lemma}
	\label{lem:testing_lb_k=2}
	Let $\varphi$ be the density function defined by $\varphi(x) = 1 + \frac{\delta}{n}\sum_{i < j} x_i x_j$. Assume $m < n/\delta^2$. Let~$\Phi : (\{-1,1\}^n)^m \to \R^{\geq 0}$ be the density associated to the distribution on $m$-tuples of strings defined as follows: First, choose $\bt$ in $\{-1,1\}^n$ uniformly; then choose $m$ strings independently from $\varphi^{+\bt}$. Let $\bone$ denote the constantly $1$ function on $(\{-1,1\}^n)^m$, the density associated to the uniform distribution. Then the $\chi^2$-divergence between $\Phi$ and $\bone$, $\|\Phi - \bone\|_2^2$, is bounded by
	\[
	\|\Phi - \bone\|_2^2 \leq O\left(\frac{m\delta^2}{n}\right).
	\] 
\end{lemma}
\begin{proof}
	We need to show that $\E[(\Phi - \bone)^2] = \E[\Phi^2] - 1 \leq O(m \delta^2/n)$. For uniform and independent $\bx^{(1)}, \dots, \bx^{(m)}$, 
	\begin{align*}
	\E[\Phi(\bx^{(1)}, \dots, \bx^{(m)})^2] &= \E_{\bx}\left[\left(\E_{\bt}\left[\prod_{i=1}^m\varphi^{+\bt}(\bx^{(i)})\right]\right)^2\right] \\
	&= \E_{\bx, \bt, \bt'}\left[\prod_{i = 1}^m\varphi^{+\bt}(\bx^{(i)})\varphi^{+\bt'}(\bx^{(i)})\right] \\
	&= \E_{\bt,\bt'}[\langle \varphi^{+\bt}, \varphi^{+\bt'}\rangle^m].
	\end{align*}
	It is a trivial fact that $\langle \varphi^{+t}, \varphi^{+t'}\rangle = \varphi * \varphi(t+t')$. Therefore 
	\[
	\E[\Phi(\bx^{(1)}, \dots, \bx^{(m)})^2] = \E[(\varphi * \varphi)^m].
	\]
	We know that $\widehat{\varphi * \varphi}(S) = \varphi(S)^2$. Therefore
	\[
	\varphi * \varphi = 1 + \frac{\delta^2}{n^2}\sum_{i < j} x_i x_j.
	\]
	
	To compute $\E[(\varphi * \varphi)^m]$, we just need to calculate the constant term of $(1 + \frac{\delta^2}{n^2}\sum_{i < j} x_i x_j)^m$ since $x_i^2 = 1$. Suppose that when expanding this out, we take $l$ terms of $x_i x_j$; we think these as $l$ (possibly parallel) edges in the complete graph on $n$ vertices. Then if these $l$ terms ``cancel out'', the associated edges form a collection of cycles, since each vertex has even degree. There are at most $n^l$ collections of cycles with $l$ edges. Considering choosing those $l$ terms (edges) in order, we get an upper bound of $(mn)^l$ for the number of ways of choosing $l$ terms of $x_ix_j$ to get canceled. Therefore we have
	\[
	\E\left[\left(1 + \frac{\delta^2}{n^2}\sum_{i \neq j} \bx_i \bx_j\right)^m\right] \leq \sum_{l = 0}^{m} (mn)^l \left( \frac{\delta^2}{n^2}\right)^l \leq \sum_{l = 0}^m \left(\frac{ m \delta^2}{n}\right)^l \leq 1 + O\left(\frac{m\delta^2}{n}\right),
	\]
	which completes the proof.
\end{proof}

Now we are ready to give the lower bound for sample complexity of testing fully uniform vs.\ far-from-pairwise uniform.

\begin{proof}[Proof of Theorem~\ref{thm:section1_testing_lowerbound}]
	If $m = o(n/\delta^2)$, by Lemma~\ref{lem:testing_lb_k=2} we have $\|\Phi - \bone\|_2^2 \leq o(1)$. Then any tester cannot distinguish, with more than $o(1)$ advantage, whether those $m$ samples are fully uniform or they are drawn from $\varphi^{+\bt}$ for some random $\bt$. 
	
	On the other hand, the proof of Theorem~\ref{thm:section1_closeness_lowerbound} shows that $\varphi$ is $\Omega(\delta)$-far from pairwise uniform, and from the Fourier characterization, we have that $\varphi^{+t}$ is pairwise uniform whenever $\varphi$ is. We can conclude that testing fully uniform versus $\delta$-far-from-pairwise-uniform needs sample complexity at least $\Omega(n/\delta^2)$.
\end{proof}

Unfortunately, we do not see an obvious way to generalize this lower bound to $k >  2$.

\section{Testing $\alpha k$-wise/fully uniform vs.\ far from $k$-wise uniform}

\label{sec:fully}

\subsection{The algorithm}

In this section we show a sample-efficient algorithm for testing whether a distribution is $\alpha k$-wise/fully uniform  or $\delta$-far from $k$-wise uniform. As a reminder, Theorem~\ref{thm:sample} indicates that the sample complexity of estimating $\W{1 \dots k}[\varphi]$ is bounded by the Fourier weight up to level $2k$. This suggests using a filter test to try to  ``kick out'' those distributions with noticeable Fourier weight up to degree $2k$.

\medskip

\textbf{Filter Test.} Draw $m_1$ samples from $\varphi$. If there exists a pair of samples $\bx, \by$ such that $\left|\sum_{i = 1}^n \bx_i \by_i\right| > t\sqrt{n}$, output ``Reject''; otherwise, output ``Accept''. 

\medskip

The Overall Algorithm is combining the Filter Test and the $\W{1\dots k}$ Estimation Test.

\medskip
\textbf{Overall Algorithm.} Do the Filter Test with $m_1$ samples and parameter $t$. If it rejects, say ``No''. Otherwise, do the $\W{1\dots k}$ Estimation Test with $m_2$ samples and $\theta = (\delta/e^k)^2$.  Say ``No'' if it outputs ``$\W{1\dots k}[\varphi] > \theta$'' and say ``Yes'' otherwise.
\medskip

Here ``Yes'' means $\varphi$ is $\alpha k$-wise/fully uniform, and ``No'' means $\varphi$ is $\delta$-far from $k$-wise uniform. We will decide the parameters $m_1, t, m_2$ in the Overall Algorithm later.

For simplicity, we denote $\overline{k} = \alpha k$. We will focus on testing $\alpha k$-wise uniform vs.\ far from $k$-wise uniform in the analysis. For fully uniformity, the analysis is almost the same, and we will discuss it at the end of this subsection.  

First of all, we will prove that if $\varphi$ is $\overline{k}$-wise uniform, it will pass the Filter Test with high probability, provided we choose $m_1$ and $t$ properly.

\begin{lemma}
	\label{lem:k'-accept}
	If $\varphi$ is $\overline{k}$-wise uniform (assuming $\overline{k}$ is even), the Filter Test will accept with probability at least .9 when $m_1^2 \leq \frac{t^{\overline{k}}}{5 \overline{k}^{\overline{k}/2}} $.
\end{lemma}
\begin{proof}
	If $\varphi$ is $\overline{k}$-wise uniform with $\overline{k}$ even, then by Markov's inequality on the $\overline{k}$-th moment, we have
	\[
	\Pr_{\substack{\bx, \by \sim \varphi \\ \text{independent}}}\left[\left|\sum_{i = 1}^n \bx_i \by_i\right| > t\sqrt{n}\right] = \Pr_{\bx, \by \sim \varphi}\left[\left(\sum_{i = 1}^n \bx_i \by_i\right)^{\overline{k}} > (t\sqrt{n})^{\overline{k}}\right] \leq \frac{\E_{\bx, \by \sim \varphi}\left[\left(\sum_{i=1}^n\bx_i 
		\by_i\right)^{\overline{k}}\right]}{t^{\overline{k}}n^{\overline{k}/2}}.
	\]
	
	When we expand $\left(\sum_{i=1}^n x_i y_i\right)^{\overline{k}}$, each term is at most degree $\overline{k}$ in $x$ or $y$. Because $\bx$ and~$\by$ are independent random variables chosen from $\overline{k}$-wise uniform distribution $\varphi$, the whole polynomial behaves the same as if $\bx$ and $\by$ were chosen from the fully uniform distribution:
	
	\begin{align*}
	\E_{\bx, \by \sim \varphi}\left[\left(\sum_{i=1}^n\bx_i 
	\by_i\right)^{\overline{k}}\right] 
	&= \E_{\bz \sim \{-1, 1\}^n} \left[\left(\sum_{i = 1}^n \bz_i\right)^{\overline{k}}\right]\\ 
	&\leq \overline{k}^{\overline{k}/2} \left(\E_{\bz \sim \{-1, 1\}^n} \left[\left(\sum_{i = 1}^n \bz_i\right)^2\right]\right)^{\overline{k}/2} \\
	&= \overline{k}^{\overline{k}/2} n^{\overline{k}/2}.
	\end{align*}
	The inequality uses hypercontractivity; see Theorem~9.21 in \cite{OD14}. Hence we have
	\[
	\Pr_{\bx, \by \sim \varphi}\left[\left|\sum_{i = 1}^n \bx_i \by_i\right| > t\sqrt{n}\right] \leq  \frac{\overline{k}^{\overline{k}/2}}{t^{\overline{k}}}.
	\]
	
	When drawing $m_1$ examples, there are at most ${m_1 \choose 2} \leq \frac12 m_1^2$ pairs. Hence by the union bound, the probability of $\varphi$ getting rejected is at most $\frac{m_1^2\overline{k}^{\overline{k}/2}}{2t^{\overline{k}}}\leq \frac1{10}$.
\end{proof}

Secondly, we claim that for any distribution $\varphi$ that does not get rejected by the Filter Test, it is close to a distribution $\varphi'$ with upper bounds on the Fourier weights of each of its levels. 

\begin{lemma}
	\label{lem:filter-low-weight}
	Any distribution $\varphi$ either gets rejected by the Filter Test with probability at least $.9$, or there exists some distribution $\varphi'$ such that:
	\begin{enumerate}
		\item $\varphi'$ and $\varphi$ are $\frac{8}{m_1}$-close in total variation distance;
		\item $\W{i}[\varphi'] \leq \frac{10^7}{m_1^2} n^i + t^in^{i/2}$ for all $i = 1, \dots, n$.
	\end{enumerate}
\end{lemma}

We will present the proof of Lemma~\ref{lem:filter-low-weight} in the next subsection. 

If $\varphi$ is not rejected by the Filter Test, Lemma~\ref{lem:filter-low-weight} tells us that it is close to some distribution $\varphi'$ with bounded Fourier weights on each of its levels. Even though we are drawing samples from $\varphi$, we can ``pretend'' that we are drawing samples from $\varphi'$ since they are close:

\begin{claim}
	\label{claim:nodiff}
	Let $m_2 \leq \frac{m_1}{200}$, and let $A(X^{(m_2)})$ be any event related to $m_2$ samples in $\{-1, 1\}^n$, $X^{(m_2)} = \{x_1, \dots, x_{m_2}\}$. Then we have
	\[
	\left|\Pr_{\bX^{(m_2)} \sim \varphi} [A(\bX^{(m_2)})] - \Pr_{\bX^{(m_2)} \sim \varphi'}[A(\bX^{(m_2)})] \right| \leq .08,
	\]
	when $\varphi$ and $\varphi'$ are $\frac{8}{m_1}$-close.
\end{claim}
\begin{proof}
	We denote by $\Phi$(respectively, $\Phi'$) the joint distribution of $m_2$ samples from $\varphi$(respectively, $\varphi'$). Then by a union bound we know that $\Phi$ and $\Phi'$ are $.04$-close, since $m_2 \frac{8}{m_1} \leq .04$. We denote $\bone[A(\bX^{(m_2)})]$ as the indicator function of event $A$ happening on $\bX^{(m_2)}$. Then we have
	\begin{align*}
	\left|\Pr_{\bX^{(m_2)} \sim \varphi}[A(\bX^{(m_2)})] -  \Pr_{\bX^{(m_2)} \sim \varphi'}[A(\bX^{(m_2)})]\right| &= \left|\sum_{X^{(m_2)}} \bone[A(X^{(m_2)})] \left(\Phi(X^{(m_2)}) - \Phi'(X^{(m_2)})\right)\right| \\
	& \leq \sum_{X^{(m_2)}} \left|\Phi(X^{(m_2)}) - \Phi'(X^{(m_2)})\right| \\
	& = 2 d_{\text{TV}}(\Phi, \Phi') \leq .08
	\end{align*}
	which completes the proof.
\end{proof}

Now we are ready to analyze the Overall Algorithm.

\begin{proof}[Proof of Theorem~\ref{thm:section1_fully_testing}]
	We discuss distinguishing between $\overline{k}$-wise uniform and $\delta$-far from $k$-wise uniform first. In the Overall Algorithm, we set the parameters $t = \left( 10^{11} (4e^4)^k \overline{k}^{\overline{k}/2} \frac{n^k}{\delta^4}\right)^{\frac1{\overline{k}-2k}}$ and $m_1 = \sqrt{\frac{t^{\overline{k}}}{5\overline{k}^{\overline{k}/2}}}$ in the Filter Test; and, we set $m_2 = \frac1{200} m_1$ and $\theta = \left(\frac{\delta}{e^k}\right)^2$ in the $\W{1 \dots k}$ Estimation test.
	
	In total we use $m_1 + m_2 = O\left(\sqrt{\frac{t^{\overline{k}}}{\overline{k}^{\overline{k}/2}}}\right)$ samples in the Overall Algorithm. By plugging in the definition of $t$ and $\overline{k} = \alpha k$, we can simplify the sample complexity to $O(\alpha)^{k/2} \cdot n^{k/2} \cdot \frac1{\delta^2} \cdot \left(\frac{n^k}{\delta^4}\right)^{1/(\alpha - 2)}$.
	
	The rest of the proof is to show the correctness of this algorithm. We discuss the two cases.
	
	\textbf{``Yes'' case:} Suppose $\varphi$ is $\overline{k}$-wise uniform. By Lemma~\ref{lem:k'-accept} we know that $\varphi$ will pass the Filter Test with probability at least .9 since $m_1^2 = \frac{t^{\overline{k}}}{5\overline{k}^{\overline{k}/2}}$. 
	
	Now $\varphi$ is $\overline{k}$-wise uniform with $\overline{k} > 2k$, which means $\widehat{\varphi}(S) = 0$ for any $1 \leq |S| \leq 2k$. Therefore by setting $\delta = \left( \frac{\theta}{e^k} \right)^2$ and $A = 1$, Theorem~\ref{thm:sample} tells us that $m_2$ samples are large enough for $\W{1\dots k}$ Estimation Test to output ``$\W{1\dots k}[\varphi] \leq \frac12 \theta$'' with probability $3/4$.
	
	The overall probability of the Overall Algorithm saying ``Yes'' is therefore at least $.9 \times \frac34 > \frac23$.
	
	\textbf{``No'' case:} Suppose $\varphi$ is $\delta$-far from $k$-wise uniform. Either $\varphi$ gets rejected by the Filter Test with probability .9, or according to Lemma~\ref{lem:filter-low-weight}, we know that there exists some distribution $\varphi'$ which is $\frac{8}{m_1}$-close to $\varphi$ and $\W{i}[\varphi'] \leq \frac{10^7}{m_1^2} n^i + t^in^{i/2}$ for all $i = 1, \dots, n$.
	
	The second stage is slightly tricky. As described in Claim~\ref{claim:nodiff}, at the expense of losing $.08$ probability, we may pretend we are drawing samples from $\varphi'$ rather than $\varphi$. Notice that $m_1^2 = \frac{t^{\overline{k}}}{5\overline{k}^{\overline{k}/2}} = \omega(n^k)$. We have
	\[
	\W{i}[\varphi'] \leq \frac{10^7}{m_1^2} n^i + t^in^{i/2} = (1 + o(1)) t^in^{i/2} \leq An^{i/2}
	\]
	for $i = 0, \dots, 2k$ with parameter $A = 1.01t^{2k}$. Then plugging $A = 1.01t^{2k}$ and $\theta = \left( \frac{\delta}{e^k}\right)^2$ into Theorem~\ref{thm:sample}, we know that the $\W{1\dots k}$ Estimation Test will say ``$\W{1\dots k}[\varphi] > \theta$'' with probability at least $\frac34$ when $\varphi'$ is $\delta$-far from $k$-wise uniform, provided we have at least $1005\frac{(2e^2)^kt^kn^{k/2}}{\delta^2}$ samples. It is easy to check $m_2 = \frac1{200}\sqrt{\frac{t^{\overline{k}}}{5\overline{k}^{\overline{k}/2}}}$ is sufficient.
	
	However, in the real algorithm we are drawing samples from $\varphi$ rather than $\varphi'$. From Claim~\ref{claim:nodiff}, we know that the estimator will accept with probability at least $\frac34 - .08 > \frac23$ when $\varphi'$ is $\delta$-far from $k$-wise uniform. Notice that $\varphi$ and $\varphi'$ are $\frac{8}{m_1}$-close, where $\frac{8}{m_1} = o\left(\frac{\delta^4}{n^k}\right)$. Hence if $\varphi$ is $\delta$-far from $k$-wise uniform, $\varphi'$ is also $\delta$-far from $k$-wise uniform, which completes the proof.
	
	Finally, for distinguishing between a distribution being fully uniform and a distribution being $\delta$-far from $k$-wise uniform, the modification we need is that in Lemma~\ref{lem:k'-accept} we use Hoeffding's inequality to get
	\[
	\Pr_{\bx, \by \sim \varphi}\left[\left|\sum_{i = 1}^n \bx_i \by_i\right| > t\sqrt{n}\right] \leq 2e^{-t^2/2},
	\]
	and then we have the constraint $m_1^2 \leq \frac1{10} e^{t^2/2}$. Following exactly the same analysis, we get the same algorithm with sample complexity $O(k)^k\cdot n^{k/2} \cdot \frac1{\delta^2} \cdot \left(\log \frac{n}{\delta}\right)^{k/2}$.
\end{proof}

\subsection{Proof of Lemma~\ref{lem:filter-low-weight}}

The rest of this section is devoted to proving Lemma~\ref{lem:filter-low-weight}. We will use the following definition in the analysis.

\begin{definition}
	For $x, y \in \{-1, 1\}^n$, we say $(x, y)$ is \emph{skewed} if $\left| \sum_{i = 1}^n x_i y_i \right|  > t\sqrt{n}$. We say that $x$ is \emph{$\beta$-bad for distribution $\varphi$} if $\Pr_{\by \sim \varphi}[(x, \by) \text{ is skewed}] > \beta$.
\end{definition}

\begin{claim}
	\label{claim:bad}
	If $\Pr_{\bx \sim \varphi}\left[\bx \text{ is $\frac8{m_1}$-bad for } \varphi\right] > \frac8{m_1}$, then $\varphi$ will be rejected by the Filter Test with probability at least $.9$.
\end{claim}
\begin{proof}
	Suppose $\Pr_{\bx \sim \varphi}\left[\bx \text{ is $\frac8{m_1}$-bad for } \varphi\right] > \frac8{m_1}$. We will divide the samples we draw for the Filter Test into two sets with size ${m_1}/2$ each. Then the probability of choosing an $\frac8{m_1}$-bad~$x$ among the first ${m_1}/2$ samples is at least
	\[
	\Pr_{\bx_1, \dots, \bx_{{m}/2} \sim \varphi}\left[\exists x \text{ $\frac8{m_1}$-bad for $\varphi$ among } \bx_1, \dots, \bx_{{m}/2}\right] > 1 - \left(1-\frac8{m_1}\right)^{{m_1}/2} \geq 1-e^{-4}.
	\]
	Now if we have such an $\frac8{m_1}$-bad $x$ among  the first ${m_1}/2$ samples, each $(x, \bx_t)$ will be skewed with probability at least $\frac8{m_1}$ for any $t = {m_1}/2+1, \dots, {m}$. Therefore
	\[
	\Pr_{\bx_{{m}/2+1}, \dots, \bx_{{m}}}[(x, \bx_t) \text{ is skewed for some } t = \frac{{m}}2+1, \dots, {m}] \geq 1 - \left(1-\frac8{m_1}\right)^{{m_1}/2} \geq 1-e^{-4}.
	\]
	Combining the two inequalities together, we know that the probability of at least one pair being skewed is at least $(1-e^{-4})^2 \geq .9$.
\end{proof}

Now we only need to consider the case when the probability of drawing a bad $x$ from $\varphi$ is very small. We want to show a stronger claim that even the probability of drawing a skewed pair from $\varphi$ is small. However this might not be true for $\varphi$ itself. Thus we look at another distribution $\varphi'$, which is defined to be $\varphi$ conditioned on outcomes being not bad. Define $\varphi'$ as
\begin{equation}
\label{eq:phi'}
\varphi'(x) = \varphi(x) \frac{ \bone\left[x \text{ not $\frac8{m_1}$-bad for $\varphi$}\right]}{1 - \Pr_{\bx \sim \varphi}\left[\bx \text{ is $\frac8{m_1}$-bad for $\varphi$}\right]}.
\end{equation}

We show that $\varphi'$ is close to $\varphi$ and that $\varphi'$ has no bad samples:

\begin{claim}
	\label{claim:phi'}
	Suppose $\varphi$ satisfies $\Pr_{\bx \sim \varphi}\left[\bx \text{ is  $\frac8{m_1}$-bad for $\varphi$}\right] \leq \frac8{{m_1}}$ and $m_1 \geq 16$. Let $\varphi'$ be defined as in \cref{eq:phi'}. Then:
	\begin{enumerate}
		\item $\varphi$ and $\varphi'$ are $\frac8{m_1}$-close;
		\item $\varphi'(x) = 0$ for any $x$ that is $\frac{16}{m_1}$-bad for $\varphi'$.
	\end{enumerate}
\end{claim}

\begin{proof}
	
	\begin{enumerate}
		\item Notice that $\varphi'(x) = 0 \leq \varphi(x)$ when $x$ is $\frac8{m_1}$-bad for $\varphi$, and $\varphi'(x) \geq \varphi(x)$ otherwise. Hence,
		\begin{align*}
		d_{\text{TV}}(\varphi, \varphi') 
		&= \frac12 \E_{\bx}[|\varphi(\bx) - \varphi'(\bx)|] \\
		&= \frac1{2^n} \sum_{\varphi'(x) < \varphi(x)} (\varphi(x) - \varphi'(x))\\ &\leq  \Pr_{\bx \sim \varphi}\left[\bx \text{ $\frac8{m_1}$-bad on $\varphi$}\right] \leq \frac8{{m_1}}.
		\end{align*}
		\item $\varphi'(x)$ is either 0 or at most $(1 + \frac{16}{{m_1}}) \varphi(x)$ given $\Pr_{\bx \sim \varphi}\left[\bx \text{ is $\frac8{m_1}$-bad for $\varphi$}\right] \leq \frac8{{m_1}}$ and $m_1 \geq 16$. Therefore if $\varphi'(x) > 0$, $x$ is not $\frac8{m_1}$-bad for $\varphi$. Hence,
		\begin{align*}
		\Pr_{\by \sim \varphi'}\left[(x, \by) \text{ is skewed}\right] 
		&\leq \left(1 + \frac{16}{{m_1}}\right) \Pr_{\by \sim \varphi}\left[(x, \by) \text{ is skewed}\right] \\
		&\leq \left(1+\frac{16}{{m_1}}\right)\frac8{m_1} \leq \frac{16}{m_1}. \qedhere
		\end{align*}
	\end{enumerate}

\end{proof}

\begin{claim}
	\label{claim:skewsmall}
	Suppose distribution $\varphi$ satisfies $\Pr_{\bx \sim \varphi}\left[\bx \text{ is $\frac8{m_1}$-bad for $\varphi$}\right] \leq \frac8{{m_1}}$. Let $\varphi'$ be defined as in \cref{eq:phi'}. If $\Pr_{\bx, \by \sim \varphi'}[(\bx, \by) \text{ is skewed}] > \frac{10^7}{m_1^2}$, then  with probability at least $.9$, $\varphi$ will be rejected by the Filter Test.
\end{claim}

We want to clarify that the constraint is about $\varphi'$, but we are drawing samples from $\varphi$ in the Filter Test.

\begin{proof}
	We only consider the first $m_1' = \frac{m_1}{200}$ samples. From Claim~\ref{claim:phi'} we know that $\varphi$ and $\varphi'$ are $\frac{8}{m_1}$-close. Therefore, we only need to show that if the samples are drawn from $\varphi'$, the probability of appearing a skewed pair among these $m_1'$ samples is at least $.98$. Then $\varphi$ will be rejected by the Filter Test with probability at least $.98 - .08 \geq .9$ according to Claim~\ref{claim:nodiff}.
	
	Define random variable $\bU_{s, t}$ to be the indicator associated with the event that $(\bx_s, \bx_t)$ is skewed, and $\bU = \sum_{1 \leq s < t \leq m_1'} \bU_{s, t}$. We need to prove that $\Pr[\bU = 0] \leq .02$. (From now on, all probabilities and expectations are based on choosing samples from distribution $\varphi'$.)  By Chebyshev's inequality, we know that $\Pr[\bU = 0] \leq \frac{\Var[\bU]}{\E[\bU]^2}$, so we need to calculate $\Var[\bU]$ and $\E[\bU]$.
	
	Denote $\mu = \Pr_{\bx, \by \sim \varphi'}[(\bx, \by) \text{ is skewed}]$. Then $\E[\bU_{s, t}] = \mu$ for any $s < t$ and hence we have
	\[
	\E[\bU] = \sum_{s < t} \E[\bU_{s, t}] = {m_1' \choose 2} \mu.
	\]
	
	It remains to calculate $\E[\bU^2]$. We can expand it as
	\[
	\E[\bU^2] = \E\left[\left(\sum_{s < t} \bU_{s, t} \right)^2\right] = \sum_{\substack{s <  t\\s' < t'}}\E[\bU_{s,t}\bU_{s', t'}].
	\]
	Similar to the proof of Lemma~\ref{lem:mean-and-var}, we discuss these expectations in three cases.
	
	\textbf{Case 1:} $|\{s, t\} \cap \{s', t'\}| = 2.$ Since $\bU_{s, t}$ is a Bernoulli random variable, we know that
	\[
	\E[\bU_{s,t}^2] = \E[\bU_{s,t}] = \mu.
	\]
	
	\textbf{Case 2:} $|\{s, t\} \cap \{s', t'\}| = 1.$ Without loss of generality we assume $s = s'$. We consider drawing $\bx_s$ first. For any fixed $x_s$ with $\varphi'(x_s) > 0$,  
	\[
	\E_{\bx_{t'}}[\bU_{s, t'}] = \Pr_{\bx_{t'}}[(x_s, \bx_{t'
	}) \text{ get skewed}] \leq \frac{16}{m_1} = \frac{2}{25m_1'},
	\] 
	where the inequality comes from Claim~\ref{claim:phi'}. Therefore,
	\[
	\E[\bU_{s, t}\bU_{s, t'}] = \E_{\bx_s, \bx_t}[\bU_{s,t} \E_{\bx_{t'}}[\bU_{s, t'}]] \leq \frac{2\mu}{25 m_1'}.
	\]
	
	\textbf{Case 3:} $|\{s, t\} \cap \{s', t'\}| = 0.$ Since $s, t, s', t'$ are all distinct, we have
	\[
	\E[\bU_{s, t}\bU_{s', t'}] = \E[\bU_{s, t}]\E[\bU_{s', t'}] = \mu^2.
	\]
	Combining these cases together, we get
	\[
	\E[\bU^2] = {m_1' \choose 2}\mu + m_1'(m_1'-1)(m_1'-2)\frac{2\mu}{25 m_1'} + {m_1' \choose 2}{m_1'-2 \choose 2} \mu^2.
	\]
	Then we have 
	\[
	\frac{\Var[\bU]}{\E[\bU]^2} = \frac{\E[\bU^2]}{\E[\bU]^2} - 1 \leq \frac{58}{25m_1'^2 \mu}.
	\] 
	By substituting $\mu \geq \frac{10^7}{m_1^2} = \frac{10^3}{4m_1'^2}$, we conclude $\Pr[\bU = 0] = \frac{\Var[\bU]}{\E[\bU]^2} \leq .02$,
	which completes the proof.
\end{proof}

Now we only need to consider those distributions $\varphi$ where their corresponding $\varphi'$ satisfies that \linebreak $\Pr_{\bx, \by \sim \varphi'}[(\bx, \by) \text{ is skewed}] \leq \frac{10^7}{m_1^2}$. This gives us an upper bound on the Fourier weight on all levels of $\varphi'$.
\begin{claim}
	\label{claim:weightsmall}
	If $\Pr_{\bx, \by \sim \varphi'}[(\bx, \by) \text{ is skewed}] \leq \frac{10^7}{m_1^2}$, then 
	\[
	\W{i}[\varphi'] \leq \frac{10^7}{m_1^2} n^i + t^i n^{i/2}
	\]
	for $i = 1, \dots, n$.
\end{claim}
\begin{proof}
	We will first show that $\W{i}[\varphi'] \leq \E_{\bx, \by \sim \varphi'}[(\sum_{j=1}^n \bx_j \by_j)^i]$. Since $(\sum_{j=1}^n x_j y_j)^i$ is a symmetric function, we can expand it as
	\[
	\left(\sum_{j=1}^n x_j y_j\right)^i  = \sum_{\substack{0\leq k \leq i\\ i-k\text{ even}}} \alpha_{k}\left(\sum_{|S| = k} x^{S} y^{S}\right),
	\]
    with positive integer coefficients $\alpha_{k}$. Notice that
	\[
	\E_{\bx, \by \sim \varphi'} \left[\sum_{|S| = k} x^{S} y^{S}\right] = \W{k}[\varphi'].
	\]
	Therefore
	\[
	\E_{\bx, \by \sim \varphi'}\left[\left(\sum_{j=1}^n \bx_j \by_j\right)^i\right] = \sum_{\substack{0\leq k \leq i \\ i-k\text{ even}}} \alpha_{k} \W{k}[\varphi'] \geq \W{i}[\varphi'].
	\]
	The last inequality holds because the $\alpha_{k}$'s are positive integers and each $\W{k}[\varphi']$ is non-negative.
	
	The rest of the proof is devoted to bounding $\E_{\bx, \by \sim \varphi'}[(\sum_{j=1}^n \bx_j \by_j)^i]$. When $(x, y)$ is not skewed, $\sum_j x_j y_j$ is at most $n$; otherwise by the definition of ``being skewed'', $\sum_j x_j y_j$ is at most $t\sqrt{n}$. Therefore,
	\[
	\E\left[\left(\sum_{j=1}^n \bx_j \by_j\right)^i\right] \leq \frac{10^7}{m_1^2} n^i + t^i n^{i/2}
	\]
	for all $i = 1, \dots, n$.
\end{proof}

Combining the above discussion, we get the proof of Lemma~\ref{lem:filter-low-weight}.

\begin{proof}[Proof of Lemma~\ref{lem:filter-low-weight}]
	We consider three cases for $\varphi$.
	
	\textbf{Case 1:} If $\Pr_{\bx \sim \varphi}\left[\bx \text{ is $\frac8{m_1}$-bad on } \varphi\right] > \frac8{m_1}$, Claim~\ref{claim:bad} tells us that $\varphi$ is rejected by the Filter Test with probability at least .9.
	
	For the remaining two cases we know that $\Pr_{\bx \sim \varphi}\left[\bx \text{ is $\frac8{m_1}$-bad on } \varphi\right] \leq \frac8{m_1}$. We construct $\varphi'$ as in \cref{eq:phi'}.
	
	\textbf{Case 2:} If $\Pr_{\bx \sim \varphi}\left[\bx \text{ is $\frac8{m_1}$-bad on } \varphi\right] \leq \frac8{m_1}$ but $\Pr_{\bx, \by \sim \varphi'}[(\bx, \by) \text{ is skewed}] > \frac{10^7}{m_1^2}$, Claim~\ref{claim:skewsmall} tells us that $\varphi$ also gets rejected with probability at least .9.
	
	\textbf{Case 3:} If $\Pr_{\bx, \by \sim \varphi'}[(\bx, \by) \text{ is skewed}] > \frac{10^7}{m_1^2}$, then according to Claim~\ref{claim:weightsmall}, $\W{i}[\varphi'] \leq \frac{10^7}{m_1^2} n^i + t^in^{i/2}$ for all $i = 1, \dots, n$. Also by Claim~\ref{claim:phi'} we know that $\varphi$ and $\varphi'$ are $\frac8{m_1}$-close.
\end{proof}

\bibliographystyle{alpha}

\bibliography{k-wise_uniform}

\end{document}